\documentclass[aps,pra,reprint,superscriptaddress,nofootinbib,longbibliography,floatfix]{revtex4-2}

\usepackage[T1]{fontenc}
\usepackage{lmodern}
\usepackage{amsmath,amssymb,mathtools,bm}
\usepackage{graphicx}
\usepackage{booktabs}
\usepackage{microtype}
\usepackage{amsthm}
\setcitestyle{numbers,sort&compress}
\usepackage[hidelinks]{hyperref}
\hypersetup{pdftitle={Sharp Passive Broadband Limits for Multiresonator Quantum Memories},pdfauthor={Maxim V. Churilov},pdfsubject={Passive quantum-memory front ends, Bode--Fano sharpness, and exact continuum certification},pdfkeywords={quantum memory front end, impedance matching, Herglotz functions, Bode--Fano bound, Sturm certificate}}
\usepackage[nameinlink,capitalize]{cleveref}

\newtheorem{theorem}{Theorem}
\newtheorem{proposition}{Proposition}

\newtheorem{corollary}{Corollary}
\theoremstyle{definition}

\newcommand{\ii}{\mathrm{i}}
\newcommand{\dd}{\mathrm{d}}
\newcommand{\RePart}{\operatorname{Re}}
\newcommand{\esssup}{\operatorname*{ess\,sup}}
\newcommand{\B}{\mathcal{B}}

\newcommand{\rmax}{\rho}

\newcommand{\Log}{\operatorname{Log}}

\makeatletter
\AtBeginDocument{%
  \@ifpackageloaded{hyperref}{\hypersetup{hidelinks}}{}%
}
\makeatother

\begin{document}

\title{Sharp Passive Broadband Limits for Multiresonator Quantum Memories}

\author{Maxim V. Churilov}
\email{churilovm1305@gmail.com}
\affiliation{Independent Researcher, Orenburg, Russia}

\date{July 16, 2026}

\begin{abstract}
Broadband quantum-memory front ends are often assessed by center-frequency impedance matching or sampled efficiency curves.  Neither supplies a continuous-band certificate, and transfer into a controlled field is not automatically storage in a finite register.  We formulate a passive one-port multiresonator front end through a positive-real spectral admittance and prove two complementary structural results.  First, a $d$-dimensional register has zero worst-case write efficiency on the full infinite-dimensional space $L^2([-B,B])$; a positive uniform guarantee for one universal linear write map therefore requires a finite-dimensional admitted signal subspace or an infinite-dimensional register.  Input-dependent mode-selective control addresses a different, nonuniversal task.  Second, for stationary diagonal passive-atom admittances the Bode--Fano value $\exp[-\pi\kappa/(2B)]$ is the exact infinite-order infimum of the band reflection norm.  We construct the nonrational outer extremal, prove its complete Herglotz representation (a bounded in-band density plus exterior point masses accumulating exponentially at the band edges) by a boundary Cauchy--Poisson argument, prove uniform closed-band approximation by finite positive-residue resonators with an explicit error bound, and prove that no finite rational design attains the infimum.  At fixed pole locations the finite-order minimax problem is quasiconvex in oscillator strengths.  Finally, exact rational arithmetic reduces a continuous reflection bound to positivity of one univariate polynomial, certified by Sturm root counting.  Two independent implementations replay the band certificates for all ten reported candidates and the Routh--Hurwitz stability and minimum-phase certificate for the headline 11-mode candidate.  In units $\kappa=2$ and $B=1$, an explicit 11-mode design satisfies $0.06410221\le\|r\|_\infty<0.0641023$, hence transfers more than $0.995890895$ of every admitted in-band photon into controlled output channels.  The same number is a memory-write guarantee precisely on subspaces for which the downstream capture map is isometric.
Every optimal dual certificate can be chosen on its contact set, with a
support bound fixed solely by the number of active moment constraints.
\end{abstract}

\maketitle

\section{Introduction}

Quantum memories are interfaces between propagating photonic or microwave modes and long-lived material or resonator degrees of freedom.  Their usefulness in networking, transduction, synchronization, and fault-tolerant architectures depends not only on a high peak efficiency, but on a well-defined signal space, a reversible storage register, and a guarantee that holds for every admitted waveform \cite{Lvovsky2009,Gorshkov2007,Heshami2016,Lei2023}.  Cavity enhancement and impedance matching have enabled high-efficiency atomic-frequency-comb and spin-ensemble protocols \cite{Afzelius2009,AfzeliusSimon2010,Moiseev2010,Sabooni2013,Hedges2010}, while multiresonator architectures provide additional spectral degrees of freedom and have been studied theoretically and experimentally in optical and superconducting platforms \cite{Moiseev2017,Moiseev2018,Perminov2019,Bao2021,Matanin2023,Perminov2023,Moiseev2026Atomic}.

Recent experiments also illustrate why peak, average, modal, and worst-case efficiencies must not be conflated.  An integrated rare-earth memory reported $80.3(7)\%$ efficiency for weak coherent pulses, $69.8(1.6)\%$ for heralded single photons, and 20 temporal modes with lower average efficiency \cite{Meng2026}; a separate interconnect-oriented memory reported more than $80\%$ efficiency across 11 spatial modes and qubit fidelities above $99\%$ \cite{Luo2026}.  These advances motivate complementary certificates that specify the admitted signal space and control the least favorable waveform.

Four distinctions are essential when such a device is advertised as broadband and near unit efficiency.  First, vanishing reflection at one frequency is a local condition; it does not control the largest reflection over a continuous band.  Second, a dense numerical grid is evidence but not an exact continuum proof unless interpolation error is bounded.  Third, a missing output photon may have entered a controlled traveling field, a recoverable register, or an uncontrolled bath.  Fourth, even a lossless transfer into controlled fields cannot be mapped isometrically from all of $L^2([-B,B])$ into a finite-dimensional register.  The scalar identity ``efficiency $=1-|r|^2$'' is consequently a storage statement only after both the output channels and the admitted signal subspace have been specified.

This work supplies a unified framework for these issues.  The central object is the self-energy $\Sigma(s)$ seen by a common input resonator.  Passivity places $\Sigma$ in the positive-real (equivalently, half-plane Herglotz) class, and the input reflection is its Cayley transform \cite{Nedic2021,Ivanenko2019}.  Classical Fano--Darlington theory already connects logarithmic matching bounds to increasingly high-order equal-ripple network synthesis \cite{Darlington1939,Bode1945,Fano1950}.  Sum-rule and Herglotz-function limits of this general type also have precedents in broadband passive cloaking \cite{CassierMilton2017}, passive approximation \cite{Ivanenko2019}, and recent Bode--Fano absorption bounds for subwavelength particles \cite{Corsaro2026}.  The contribution here is narrower and explicit: exact sharpness is proved inside the diagonal finite positive-residue class, including a complete measure and a quantitative closed-band atomic approximation, and finite candidates are certified in exact arithmetic.  A controlled-output dilation converts the reflection defect into transfer probability; a further capture map is required before that transfer can be called memory write.  We prove the associated finite-register obstruction, formulate the correct finite-dimensional signal-subspace criterion, and then solve the infinite-order stationary minimax value.

The sharpness proof is constructive.  An outer Schur function has constant modulus on a slightly enlarged band and exhausts the logarithmic Bode--Fano budget.  Its inverse Cayley transform has an explicit positive Herglotz measure, established here by a direct boundary Cauchy--Poisson argument rather than by invoking a representation theorem whose positive-reality hypothesis is the very statement at issue.  Lorentzian regularization followed by positive atomic quadrature produces admissible finite-resonator approximants, while rationality rules out attainment at any finite order.  Thus the Bode--Fano floor is not merely a lower bound: it is the exact infimum of the finite-atom closure.

The complementary computational advance is an independently rerunnable continuum certificate.  For a rationalized design, $|r(\ii\omega)|^2$ is a rational function of $x=\omega^2$.  The statement $|r(\ii\omega)|<\bar\rho$ for all $|\omega|\le B$ is equivalent to strict positivity of one polynomial on a compact interval.  Exact Sturm arithmetic proves that positivity without a grid, and exact Routh--Hurwitz determinants establish stability and minimum phase.  Applied to an 11-mode candidate obtained by nonlinear synthesis, this yields a reflection bracket narrower than $10^{-7}$.

We deliberately separate five levels of assertion.  Controlled-port transfer is fixed by scattering, whereas memory write additionally depends on a capture map and its signal domain.  The Bode--Fano floor is the exact infinite-order infimum under the stated stationary passive assumptions.  The fixed-support oscillator-strength problem is quasiconvex.  The reported rational designs have exact continuum upper certificates.  Global optimality at each finite movable-pole order, and the rate at which those optima approach the floor, remain open approximation-theoretic questions.  This separation avoids promoting equiripple plots or unsuccessful global searches into theorems.

\section{Passive one-port model and the storage qualification}
\label{sec:model}

\subsection{Reduced input--output model}

Let $a$ be a common resonator coupled to a one-dimensional input field with rate $\kappa$ and to internal linear modes.  In a rotating frame, elimination of the internal modes gives
\begin{equation}
 \left[s+\frac{\kappa+\gamma_c}{2}+\Sigma(s)\right]a(s)
 =\sqrt{\kappa}\,A_{\mathrm{in}}(s)+\text{noise inputs},
 \label{eq:reduced}
\end{equation}
with $A_{\mathrm{out}}=A_{\mathrm{in}}-\sqrt{\kappa}a$.  The coherent reflection amplitude is
\begin{equation}
 r(s)=\frac{s-\kappa/2+\gamma_c/2+\Sigma(s)}
 {s+\kappa/2+\gamma_c/2+\Sigma(s)}.
 \label{eq:reflection-general}
\end{equation}
A microscopic resonator--ensemble realization leads to nested Stieltjes-type self-energies.  For synthesis and certification we use the finite atomic class
\begin{equation}
 \Sigma_N(s)=\sum_{n=1}^{N}\frac{w_n}{s+\Gamma_n+\ii\nu_n},
 \qquad w_n>0,\quad \Gamma_n>0,
 \label{eq:atomicSigma}
\end{equation}
with conjugate pairs included so that the total transfer function has real coefficients.  Each term is positive real in $\RePart s>0$ because
\begin{equation}
 \RePart\frac{w_n}{s+\Gamma_n+\ii\nu_n}
 =\frac{w_n(\RePart s+\Gamma_n)}{|s+\Gamma_n+\ii\nu_n|^2}>0.
\end{equation}
More general passive networks are covered by matrix positive-real and Darlington realization theory \cite{Gardiner1985,GoughZhang2015,Nurdin2016,Darlington1939,Dewilde1999}; the diagonal atomic class is used because it maps directly to buildable detuned resonators and admits exact rational certificates.

\begin{proposition}[Positive-real--Schur correspondence]
For $\gamma_c=0$, if $\Sigma$ is analytic and positive real in $\RePart s>0$, then
\begin{equation}
 r(s)=\frac{s-\kappa/2+\Sigma(s)}{s+\kappa/2+\Sigma(s)}
\end{equation}
is analytic and satisfies $|r(s)|<1$ in the open right half-plane.  Conversely, a target $r$ belongs to this memory class precisely when
\begin{equation}
 \Sigma_r(s)=\frac{\frac{\kappa}{2}[1+r(s)]-s[1-r(s)]}{1-r(s)}
 \label{eq:inverse-cayley}
\end{equation}
is positive real and admits the chosen physical realization.
\end{proposition}

\begin{proof}
Set $z=s+\Sigma(s)$.  Then $\RePart z>0$, and $(z-\kappa/2)/(z+\kappa/2)$ maps the right half-plane into the unit disk.  Solving the Cayley transform for $\Sigma$ gives \cref{eq:inverse-cayley}.
\end{proof}

\subsection{Controlled-output dilation and memory qualification}

Positive realness certifies passivity, not memory.  The widths $\Gamma_n$ in \cref{eq:atomicSigma} first define controlled Markov output channels, not long-lived registers.  Introduce modes $b_n$ obeying
\begin{align}
 \dot a&=-\frac{\kappa}{2}a-\ii\sum_n g_n b_n+\sqrt{\kappa}A_{\mathrm{in}},\label{eq:dil-a}\\
 \dot b_n&=-(\Gamma_n+\ii\nu_n)b_n-\ii g_n^*a+\sqrt{2\Gamma_n}\,S_{n,\mathrm{in}},\label{eq:dil-b}\\
 S_{n,\mathrm{out}}&=S_{n,\mathrm{in}}-\sqrt{2\Gamma_n}\,b_n.
 \label{eq:dil-out}
\end{align}
Here $|g_n|^2=w_n$.  Let $\mathcal K_{\mathrm{ctl}}$ denote the admitted one-photon subspace of the outgoing fields $S_{n,\mathrm{out}}$.  A downstream memory is described on the one-photon sector by a contraction
\begin{equation}
 C_{\mathrm{cap}}:\mathcal K_{\mathrm{ctl}}\longrightarrow\mathcal H_{\mathrm{mem}},
 \label{eq:ucap}
\end{equation}
where $\mathcal H_{\mathrm{mem}}$ is a long-lived register.  It is an isometry only on signal subspaces captured without loss.  Uncontrolled intrinsic loss must be represented by additional bath ports and excluded from the write probability.  \Cref{fig:interface} summarizes this qualification chain.

\begin{figure*}[t]
 \includegraphics[width=0.96\textwidth]{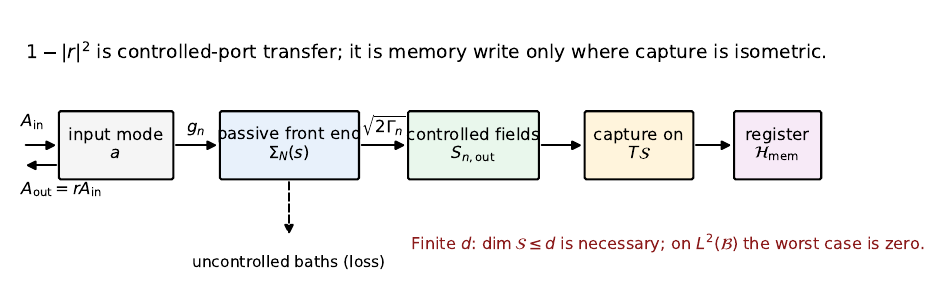}
 \caption{Operational qualification of the passive interface.  The rational admittance fixes scattering from the signal port into controlled output fields.  Reflection defect becomes memory-write probability only on a specified signal subspace that the capture map stores isometrically; dashed channels represent uncontrolled loss and do not count as storage.  A finite-dimensional register cannot satisfy this condition on all of $L^2(\B)$.}
 \label{fig:interface}
\end{figure*}

For vacuum controlled-channel inputs, the transfer amplitude from $A_{\mathrm{in}}$ to the $n$th controlled output is
\begin{equation}
 t_n(s)=\frac{\ii\sqrt{2\kappa\Gamma_n w_n}}
 {(s+\Gamma_n+\ii\nu_n)[s+\kappa/2+\Sigma_N(s)]}.
 \label{eq:t-transfer}
\end{equation}
Direct substitution on the imaginary axis gives
\begin{equation}
 |r(\ii\omega)|^2+\sum_n|t_n(\ii\omega)|^2=1.
 \label{eq:unitarity}
\end{equation}

\begin{theorem}[Exact controlled-channel transfer certificate]
Let $f\in L^2(\B)$ be a normalized one-photon spectrum.  Suppose that all non-signal inputs in \cref{eq:dil-a,eq:dil-b,eq:dil-out} begin in vacuum and that no uncontrolled bath is present.  The probability transferred from the signal port into the controlled output channels is
\begin{equation}
 \eta_{\mathrm{ctl}}[f]=1-\int_{\B}|r(\ii\omega)|^2|f(\omega)|^2\dd\omega.
 \label{eq:write-exact}
\end{equation}
Consequently,
\begin{equation}
 \inf_{\substack{\|f\|_2=1\\\operatorname{supp}f\subseteq\B}}
 \eta_{\mathrm{ctl}}[f]
 =1-\esssup_{\omega\in\B}|r(\ii\omega)|^2.
 \label{eq:minimax-operational}
\end{equation}
For a closed admitted signal subspace $\mathcal S\subseteq L^2(\B)$, the same expression is the memory-write probability on $\mathcal S$ if and only if $C_{\mathrm{cap}}$ is isometric on the controlled-output image of $\mathcal S$.
\end{theorem}

\begin{proof}
The signal-plus-controlled-field scattering map is an isometry.  Equation \eqref{eq:unitarity} gives the frequency-resolved balance, and integration against $|f|^2$ gives the probability transferred to $\mathcal K_{\mathrm{ctl}}$.  Multiplication by $|r|^2$ is a bounded operator on $L^2(\B)$ whose operator norm is its essential supremum, yielding \cref{eq:minimax-operational}.  A subsequent contraction preserves this probability for every $f\in\mathcal S$ exactly when it preserves the norm of every vector in the corresponding controlled-output image.
\end{proof}

\begin{proposition}[Finite-register obstruction]
Let $\mathcal S$ be an admitted one-photon signal space and let the complete linear write map be $W:\mathcal S\to\mathcal H_{\mathrm{mem}}$.  If $\dim\mathcal H_{\mathrm{mem}}=d<\dim\mathcal S$, then
\begin{equation}
 \inf_{\|f\|=1,\ f\in\mathcal S}\|Wf\|^2=0.
 \label{eq:finite-rank-obstruction}
\end{equation}
In particular, every finite-dimensional register has zero worst-case write efficiency on $\mathcal S=L^2(\B)$, irrespective of the reflection spectrum.  Isometric capture of a $K$-dimensional signal subspace requires $d\ge K$.
\end{proposition}

\begin{proof}
The rank of $W$ is at most $d$.  If $\dim\mathcal S>d$, rank--nullity gives a nonzero $f\in\ker W$; after normalization it attains zero write probability.  The statement includes the infinite-dimensional case.
\end{proof}

The proof is elementary linear algebra, and we deliberately state the result as a proposition rather than a theorem; its value is operational, because broadband-memory efficiency claims routinely leave the admitted signal space unspecified, and \cref{eq:finite-rank-obstruction} shows that the omission is not innocent.  The transfer theorem is the reason to optimize $\|r\|_\infty$, rather than a center value or an average, but \cref{eq:finite-rank-obstruction} fixes its interpretation.  Without a subspace-qualified capture isometry, \cref{eq:unitarity} certifies loading into controlled propagating channels, not storage.  If some $\Gamma_n$ denotes irreversible material decay, \cref{eq:write-exact} is only an absorption identity.  The full-band number reported below is therefore an unconditional controlled-port transfer guarantee and a conditional memory guarantee.

\section{Stationary passive limitations}
\label{sec:limits}

We first isolate what no finite stationary passive architecture can do.

\begin{theorem}[No finite exact broadband matching]
Let $\Sigma_N$ be a finite rational positive-real self-energy that vanishes as $s\to\infty$.  Then the reflection \cref{eq:reflection-general} with $\gamma_c=0$ cannot vanish on any open interval of the real-frequency axis.
\end{theorem}

\begin{proof}
If the rational function $r$ has boundary value zero on an open interval and no poles there, analytic continuation implies $r\equiv0$.  Its numerator would then require $\Sigma_N(s)=\kappa/2-s$, which neither vanishes at infinity nor remains positive real for large positive $s$.
\end{proof}

The stronger quantitative limitation follows from the logarithmic area of a passive reflection coefficient.  Let $\B=[-B,B]$ and assume that $r$ is rational, Schur in the right half-plane (boundedness already excludes imaginary-axis poles), obeys $r(\infty)=1$, and has expansion
\begin{equation}
 r(s)=1-\frac{\kappa}{s}+O(s^{-2}).
 \label{eq:asymptotic-r}
\end{equation}
The usual half-plane Poisson--Jensen argument gives
\begin{equation}
 \int_{-\infty}^{\infty}\log\frac{1}{|r(\ii\omega)|}\dd\omega
 =\pi\kappa-2\pi\sum_{z_k\in\mathbb{C}_+}\RePart z_k
 \le \pi\kappa,
 \label{eq:bode-area}
\end{equation}
where the sum is over open right-half-plane zeros, counted with multiplicity \cite{Bode1945,Fano1950}.  Imaginary-axis zeros are permitted: for a rational $r\not\equiv0$ they contribute integrable logarithmic singularities, and \cref{eq:bode-area} follows by a limiting contour argument detailed in the Supplemental Material.  The equality condition in \cref{eq:bode-area} is therefore a property of a particular minimum-phase reflection, not a universal property of every optimum.

\begin{corollary}[Bode--Fano minimax floor]
Under the preceding assumptions, any stationary passive memory satisfying $|r(\ii\omega)|\le\rmax$ for $|\omega|\le B$ obeys
\begin{equation}
 \boxed{\quad \rmax\ge \exp\left(-\frac{\pi\kappa}{2B}\right).\quad}
 \label{eq:BF-floor}
\end{equation}
Thus the worst-case controlled-port transfer, and therefore any downstream write efficiency, is bounded by
\begin{equation}
 \eta_{\mathrm{ctl}}^{\mathrm{wc}}
 \le 1-\exp\left(-\frac{\pi\kappa}{B}\right).
\end{equation}
\end{corollary}

\begin{proof}
The integral over the design band in \cref{eq:bode-area} is at least $2B\log(1/\rmax)$.  Combining the two inequalities gives \cref{eq:BF-floor}.
\end{proof}

The logarithmic bound leaves open whether the floor is compatible with the restricted positive-residue self-energies in \cref{eq:atomicSigma}.  The following result resolves that question.  Let $\mathfrak A$ be the union over all finite, real-symmetric atomic self-energies with $w_n,\Gamma_n>0$.

\begin{theorem}[Sharpness in the atomic closure and finite-order nonattainment]
For the band $[-B,B]$ and fixed $\kappa>0$,
\begin{equation}
 \inf_{\Sigma\in\mathfrak A}
 \left\|\frac{s-\kappa/2+\Sigma(s)}{s+\kappa/2+\Sigma(s)}
 \right\|_{L^\infty(\ii[-B,B])}
 =\exp\!\left(-\frac{\pi\kappa}{2B}\right).
 \label{eq:sharp-infimum}
\end{equation}
Every finite rational member satisfies a strict inequality.  For any $L>B$, define $\beta=\kappa/(2L)$ and
\begin{equation}
 r_L(s)=\exp\!\left[\ii\beta
 \Log\!\left(\frac{s+\ii L}{s-\ii L}\right)\right],
 \qquad \RePart s>0,
 \label{eq:outer-extremal}
\end{equation}
where the analytic logarithm tends to zero at infinity.  Then
\begin{equation}
 |r_L(\ii\omega)|=\rho_L:=e^{-\pi\kappa/(2L)}
 \quad (|\omega|<L),
 \label{eq:outer-flat}
\end{equation}
and the inverse Cayley self-energy
\begin{equation}
 \Sigma_L(s)=\frac{\kappa}{2}\frac{1+r_L(s)}{1-r_L(s)}-s
 \label{eq:outer-sigma}
\end{equation}
is positive real, vanishes at infinity, admits the complete Herglotz representation constructed in the Supplemental Material, and is a uniform boundary limit of self-energies in $\mathfrak A$ on the closed band $[-B,B]$ whenever $B<L$.
\end{theorem}

\begin{proof}[Proof outline]
The M\"obius quotient in \cref{eq:outer-extremal} maps the open right half-plane onto the upper half-plane, so the analytic logarithm has imaginary part in $(0,\pi)$ there and $\rho_L<|r_L|<1$; the boundary argument gives \cref{eq:outer-flat}, while $r_L(s)=1-\kappa/s+O(s^{-2})$.  Because $\log|r_L|$ equals $-\pi\beta$ times the harmonic measure of the band segment, it coincides with the Poisson integral of its own boundary values, so $r_L$ has no inner factor and is outer \cite{Garnett2007}.  The Supplemental Material proves the complete Herglotz representation of \cref{eq:outer-sigma} by a boundary Cauchy--Poisson argument that does not presuppose positive reality: the measure consists of a bounded smooth density on $(-L,L)$ and positive point masses at a symmetric exterior sequence accumulating exponentially at $\pm L$, with finite total mass $L^2/3+\kappa^2/12$.  Shifting this measure by a linewidth $\Gamma>0$ and applying positive symmetric quadrature produces exactly the atoms in \cref{eq:atomicSigma}, with the explicit reflection error bound $4\varepsilon/(\kappa-2\varepsilon)$ for a uniform band self-energy error $\varepsilon<\kappa/2$.  Sending the quadrature error to zero, then $\Gamma\downarrow0$, then $L\downarrow B$---in this order, each step uniform on $[-B,B]$---proves the upper direction in \cref{eq:sharp-infimum}; the Bode--Fano corollary supplies the lower direction.

If a finite rational design attained the floor, equality would hold at every step of the area argument.  Hence $|r|$ would equal the floor almost everywhere in-band and $1$ almost everywhere out-of-band.  The rational identity $N(\ii\omega)N(-\ii\omega)=D(\ii\omega)D(-\ii\omega)$ on any exterior interval would then hold identically, forcing $|r|=1$ on the entire axis, a contradiction.
\end{proof}

The assumptions matter.  Active elements, nonstationary coupling, prior knowledge of a single temporal mode, or non-Foster loading can evade the stationary passive formulation, but they consume a different resource \cite{Dilley2012,Shlivinski2018,Greggio2026Optimal}.  The theorem determines the infinite-order value; it does not determine the optimum at a prescribed finite number of atoms.

\section{Passive minimax synthesis}
\label{sec:synthesis}

We normalize $\kappa=2$ and $B=1$.  The Bode--Fano floor is then
\begin{equation}
 \rmax_{\mathrm{BF}}=e^{-\pi}=0.043213918\ldots.
\end{equation}
We restrict to a real-symmetric $N=1+2M$ mode family,
\begin{equation}
 \Sigma_N(s)=\frac{w_0}{s+\gamma_0}
 +\sum_{j=1}^{M}w_j\left[
 \frac{1}{s+\gamma_j+\ii\delta_j}
 +\frac{1}{s+\gamma_j-\ii\delta_j}\right],
 \label{eq:symmetric-class}
\end{equation}
with all parameters positive.  The design objective is
\begin{equation}
 \min_{\bm\gamma,\bm\delta,\bm w>0}
 \max_{|\omega|\le1}|r(\ii\omega)|.
 \label{eq:global-nonlinear}
\end{equation}
The movable-pole problem is nonconvex.  For fixed poles, however, feasibility at a prescribed $0<\rmax<1$ is convex in the weights.

\begin{theorem}[Fixed-support quasiconvexity]
Fix $\gamma_0$ and $(\gamma_j,\delta_j)$.  Let $h(\omega)^T\bm w=\Sigma_N(\ii\omega)$, $A=\ii\omega-1$, and $C=\ii\omega+1$.  For fixed $\rmax<1$, the condition $|A+h^T\bm w|\le\rmax|C+h^T\bm w|$ is the convex quadratic inequality
\begin{multline}
 (1-\rmax^2)|h^T\bm w|^2
 +2\RePart\{(A-\rmax^2 C)^*h^T\bm w\}\\
 +|A|^2-\rmax^2|C|^2\le0,
 \label{eq:convex-weight}
\end{multline}
with $\bm w\ge0$.  Therefore the exact continuous-frequency feasible set, obtained by intersecting this convex constraint over all $\omega\in[-1,1]$, is convex.  Minimization over $\rmax$ is a quasiconvex semi-infinite program; a finite-grid implementation is a numerical approximation that requires a posteriori continuum verification.
\end{theorem}

\begin{proof}
Squaring the modulus inequality and collecting terms gives \cref{eq:convex-weight}.  For real $\bm w$, $|h^T\bm w|^2=\bm w^T\RePart(h^*h^T)\bm w$, and $\RePart(h^*h^T)$ is positive semidefinite because this quadratic form equals a squared modulus.  Multiplication by $1-\rmax^2>0$ therefore preserves convexity.  Intersection over frequencies and with the nonnegative orthant preserves convexity.
\end{proof}

This is a structured semi-infinite program in the sense of \cite{Hettich1993}; exactness refers to the continuum formulation, not to a finite grid used by a numerical solver.

We use two complementary discovery routes.  The finite-order sequence reported below was obtained by soft-maximum continuation, warm starts in $N$, multistart local optimization, and fixed-support convex polishing.  The sharpness theorem supplies a deterministic alternative: discretize the explicit Herglotz measure at a small positive linewidth, then polish the positive weights on that support.  Both routes are synthesis heuristics at fixed finite order, not proofs of global movable-pole optimality.  The measure-seeded route is itself end-to-end certifiable with no optimization at any stage: the Supplemental Material reports uniform-quadrature designs at linewidths $\Gamma=0.14,0.10,0.06$ whose exact rational Bernstein certificates give band bounds $0.22320$, $0.17942$, and $0.13021$ at $34$, $48$, and $78$ modes, converging as predicted by the a priori bound $\rho_L+4\varepsilon/(\kappa-2\varepsilon)$ but far less economically than the optimized sequence.  Ripple counts and independent searches are useful diagnostics, but the classical alternation theorem cannot be imported without proving the required varisolvence or Haar property for this nonlinear positive-real family.  Modern rational-minimax optimality conditions provide a promising dual and second-order framework, but adapting them to positivity, movable stable poles, and continuum constraints remains open \cite{Zhang2026Minimax}.

The one-mode problem is analytically solvable.  The choice
\begin{equation}
 \gamma_0=\sqrt{2},\qquad w_0=2
\end{equation}
gives
\begin{equation}
 |r(\ii\omega)|^2=
 \frac{\omega^4-\omega^2+6-4\sqrt2}
 {\omega^4-\omega^2+6+4\sqrt2},
\end{equation}
whose maxima on $[-1,1]$ occur at $\omega=0,\pm1$ and equal $(3-2\sqrt2)^2$.  A direct lower argument, given in the Supplemental Material, shows
\begin{equation}
 \rmax_1=3-2\sqrt2=0.171572875\ldots.
 \label{eq:N1exact}
\end{equation}

Figure~\ref{fig:spectra} shows the numerically synthesized sequence.  The finite-order curves become increasingly equiripple-like, but the formal result reported below is the exact continuum enclosure for each explicit rationalized candidate, not a global lower certificate of the same order.

\begin{figure}[tb]
 \includegraphics[width=\columnwidth]{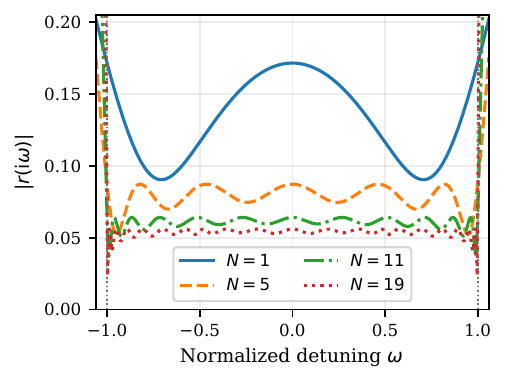}
 \caption{Reflection magnitude for four passive symmetric candidates.  Vertical dotted lines mark the design band.  The spectra flatten as the number of controlled internal modes increases.}
 \label{fig:spectra}
\end{figure}

\section{Exact computer-assisted continuum certificate}
\label{sec:exactcert}

A grid evaluation can miss a narrow peak.  Derivative bounds provide one remedy, but rational systems admit a stronger exact reduction.  Write the rationalized reflection as
\begin{equation}
 r(s)=\frac{N(s)}{D(s)},\qquad N,D\in\mathbb{Q}[s].
\end{equation}
Because the coefficients are real,
\begin{equation}
 |N(\ii\omega)|^2=n(\omega^2),\qquad
 |D(\ii\omega)|^2=d(\omega^2),
\end{equation}
for $n,d\in\mathbb{Q}[x]$.  For a rational proposed upper bound $\bar\rmax$, define
\begin{equation}
 H_{\bar\rmax}(x)=\bar\rmax^2 d(x)-n(x).
 \label{eq:Hpoly}
\end{equation}
If $D$ has no imaginary-axis zero, then
\begin{equation}
\begin{aligned}
 H_{\bar\rmax}(x)>0\quad&(0\le x\le1)\\[-2pt]
 &\Longleftrightarrow\quad
 |r(\ii\omega)|<\bar\rmax\quad(|\omega|\le1).
\end{aligned}
 \label{eq:equiv-H}
\end{equation}

\begin{theorem}[Exact rational continuum certificate]
Let $r=N/D$ have real rational coefficients, let $\bar\rmax>0$ be rational, and assume that $D$ is strictly Hurwitz.  Then the strict band inequality in \cref{eq:equiv-H} is certified by the finite exact conditions
\begin{align}
 H_{\bar\rmax}(0)&>0,\qquad H_{\bar\rmax}(1)>0,\nonumber\\
 \#\{x\in[0,1]:H_{\bar\rmax}(x)=0\}&=0.
 \label{eq:sturm-conditions}
\end{align}
The root count is evaluated exactly from a Sturm sequence.  Strict Hurwitz stability is certified by a positive leading coefficient and positive leading principal Hurwitz determinants of $D$; applying the same test to $N$ certifies minimum phase.
\end{theorem}

\begin{proof}
A continuous real polynomial with positive endpoints and no root on the interval cannot change sign or touch zero, proving \cref{eq:equiv-H}.  Sturm's theorem computes the number of distinct real roots in a rational interval using sign variations of an exact Euclidean remainder sequence \cite{Basu2006}.  The final statement is the Routh--Hurwitz criterion \cite{Gantmacher1959}.
\end{proof}

The accompanying script constructs $N,D,n,d,H$ in exact rational arithmetic, records coefficient hashes, counts roots, and produces a rigorous lower sample by evaluating $|r|^2$ at a rational frequency.  Neither Hurwitz test is needed for the band bound itself---the Sturm conditions alone certify \cref{eq:equiv-H}; the determinant tests add the separate physical statements of strict stability and of minimum phase, the latter being what upgrades the log-area inequality \cref{eq:bode-area} to an equality for this device.  Decimal parameters are not treated as uncertain real numbers: they define the explicit rational device being certified.  Fabrication uncertainty is a separate question addressed in \cref{sec:robustness}.  The complete derivations, rational parameters, proof-record fields, and reproduction protocol are provided in the Supplemental Material.

\paragraph*{Computational reproducibility.}
Every reported symbolic and numerical certificate can be regenerated from the archived rational parameters.  A standard-library verifier, implemented independently of the primary computer-algebra code, rebuilds the certificate polynomials and reproduces all thirty coefficient hashes, every Sturm count, the endpoint signs, and the reported Hurwitz signs.  The Supplemental Material specifies the proof records, dependencies, and reproduction commands.

\subsection{Headline 11-mode certificate}

The certified 11-mode design has the form \cref{eq:symmetric-class}; its parameters are listed in \cref{tab:n11}.  All entries shown to 12 decimal places are interpreted exactly as decimal rationals by the certificate script.  The residue $w_j$ is the squared coherent coupling in the normalized model, and each nonzero $\delta_j$ represents two conjugately detuned modes.

\begin{table}[tb]
\caption{Rationalized 11-mode design in units $\kappa/2=B=1$.  The central row represents one mode; each other row represents the pair $\pm\delta_j$.}
\label{tab:n11}
\begin{ruledtabular}
\resizebox{\columnwidth}{!}{%
\begin{tabular}{crrr}
 group & {$\gamma_j$} & {$\delta_j$} & {$w_j$}\\
\colrule
0 & 0.311158924086 & 0 & 0.157703925598\\
1 & 0.289160164239 & 0.304079422147 & 0.141094382446\\
2 & 0.233411007066 & 0.571893057110 & 0.103373815417\\
3 & 0.163312213040 & 0.781631334355 & 0.063754097299\\
4 & 0.093506841885 & 0.924941677224 & 0.032007238487\\
5 & 0.032916047574 & 0.999323092623 & 0.010263869247\\
\end{tabular}
}
\end{ruledtabular}
\end{table}

For this rational system, $N$ and $D$ have degree 12, $H$ has degree 12, and exact arithmetic gives
\begin{equation}
 \boxed{\quad
 0.06410221\le \rmax_{11}<0.0641023.
 \quad}
 \label{eq:n11bracket}
\end{equation}
The lower number follows from one exact rational-frequency evaluation and is therefore a lower bound on the supremum.  For the upper number, the Sturm count of $H_{0.0641023}$ on $[0,1]$ is zero and both endpoint values are positive.  All 12 Hurwitz determinants of the reflection denominator are positive, as are all 12 determinants of the numerator.  Hence the rationalized device is stable and minimum phase.  Minimum phase, together with \cref{eq:asymptotic-r}, also implies exact saturation of the full log-area identity \cref{eq:bode-area} for this particular device.  \Cref{fig:n11} shows the certified spectrum against the exact bound.

By \cref{eq:minimax-operational}, the exact upper certificate implies the controlled-port transfer guarantee
\begin{equation}
 \eta_{\mathrm{ctl}}^{\mathrm{wc}}
 >1-(0.0641023)^2
 =0.99589089513471.
 \label{eq:n11eta}
\end{equation}
On a signal subspace captured isometrically, this is also the write guarantee.  If write and read on that subspace are reciprocal, lossless, and mode matched, the corresponding worst-case round-trip probability exceeds
\begin{equation}
 [1-(0.0641023)^2]^2=0.9917986750122.
 \label{eq:n11roundtrip}
\end{equation}
Equation~\eqref{eq:n11roundtrip} does not include storage-time dephasing or errors from an uncharacterized read pulse.

\begin{figure}[tb]
 \includegraphics[width=\columnwidth]{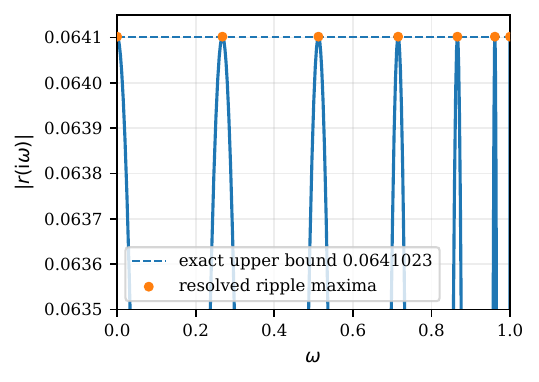}
 \caption{Positive-frequency half of the 11-mode reflection spectrum.  The dashed line is the exact continuum upper bound.  Markers show numerically resolved ripple maxima; they are diagnostics, whereas the bound is proved by exact polynomial arithmetic.}
 \label{fig:n11}
\end{figure}

The same procedure certifies all rationalized candidates in \cref{tab:sequence}.  Bracket widths are below $3\times10^{-7}$.  These are exact brackets for the norms of the listed candidate systems, not brackets for the unknown order-$N$ global optima.  Only the $N=1$ row is also proved globally optimal within its order.

\section{Finite-dimensional signal subspaces and reciprocal readout}
\label{sec:alphabet}

A finite physical register can support a finite-dimensional signal or code subspace, even though that subspace contains a continuum of coherent superpositions.  For orthonormal spectra $f_1,\ldots,f_K$, define the controlled-transfer Gram matrix
\begin{equation}
 G_{jk}=\int_{\B}f_j^*(\omega)
 [1-|r(\ii\omega)|^2]f_k(\omega)\dd\omega.
 \label{eq:gram}
\end{equation}
For a normalized coefficient vector $\bm c$, the controlled-port transfer probability is $\bm c^\dagger G\bm c$.  Therefore
\begin{equation}
 \eta_{\mathrm{ctl},K}^{\mathrm{wc}}=\lambda_{\min}(G).
 \label{eq:gramwc}
\end{equation}
If the capture map is isometric on the image of this alphabet, the same eigenvalue is its worst-case write efficiency.  Such an isometry requires a register dimension at least $K$ by \cref{eq:finite-rank-obstruction}.  The alphabet bound can be less conservative than the full-band transfer certificate because the basis may avoid frequencies at which $|r|$ is largest.  Conversely, reporting only diagonal efficiencies $G_{jj}$ is insufficient; coherent superpositions probe off-diagonal terms.

If the read process is the exact reciprocal adjoint of the write isometry after ideal rephasing, then the signal-space amplitude map is $G$ and the worst-case round-trip probability is $\lambda_{\min}(G)^2$.  In a real device, storage-time dephasing, control errors, thermal occupation, and added noise must be included as independent channel parameters.  Singular values alone cannot certify closeness to the identity when an uncalibrated unitary phase error is present.  This observation corrects a common but invalid replacement of a full channel-distance analysis by a bound involving only the smallest singular value.

\begin{table}[tb]
\caption{Exact continuum brackets for rationalized candidates.  The Bode--Fano floor is $e^{-\pi}=0.043213918\ldots$.}
\label{tab:sequence}
\begin{ruledtabular}
\begin{tabular}{rcc}
$N$ & certified lower bound & certified upper bound\\
\colrule
1  & 0.171572875 & 0.1715731\\
3  & 0.109370452 & 0.1093707\\
5  & 0.087250390 & 0.0872506\\
7  & 0.075745737 & 0.0757460\\
9  & 0.068756643 & 0.0687569\\
11 & 0.06410221 & 0.0641023\\
13 & 0.060833252 & 0.0608335\\
15 & 0.058375564 & 0.0583758\\
17 & 0.056526621 & 0.0565269\\
19 & 0.056027561 & 0.0560278\\
\end{tabular}
\end{ruledtabular}
\end{table}

\section{Robustness and experimental translation}
\label{sec:robustness}

Let
\begin{equation}
 D(\omega)=\ii\omega+\kappa/2+\Sigma(\ii\omega),\qquad
 d=\inf_{\omega\in\B}|D(\omega)|.
\end{equation}
For a perturbation $\delta\Sigma$ satisfying
$\epsilon=\sup_{\B}|\delta\Sigma|<d$, direct subtraction of the two Cayley transforms gives
\begin{equation}
 |r_{\Sigma+\delta\Sigma}-r_\Sigma|
 \le \frac{\kappa\epsilon}{d(d-\epsilon)}.
 \label{eq:robust-bound}
\end{equation}
Thus, if $\|r_\Sigma\|_\infty\le\rmax$,
\begin{equation}
 \eta_{\mathrm{ctl}}^{\mathrm{wc}}(\Sigma+\delta\Sigma)
 \ge 1-\left[\rmax+\frac{\kappa\epsilon}{d(d-\epsilon)}\right]^2.
 \label{eq:robust-eta}
\end{equation}
This is a deterministic certificate once a calibration model provides $\epsilon$.  Parameter-wise tolerances can be converted to $\epsilon$ by interval arithmetic or a triangle bound on the rational atoms.

Figure~\ref{fig:robust} shows a separate, fully reproducible Monte Carlo sensitivity study.  Every positive 11-mode parameter is perturbed independently according to the mean-preserving log-normal model $p'=p\exp(\sigma Z-\sigma^2/2)$ with $Z\sim\mathcal N(0,1)$; 1200 samples per spread, a 12001-point grid on $[0,1]$, and seed 20260711 are used.  This is not a confidence statement about any fabrication process unless that distribution is experimentally justified.  Its purpose is diagnostic: the nominal minimax design is sharp, and robust synthesis or postfabrication retuning is likely to be important.

\begin{figure}[tb]
 \includegraphics[width=\columnwidth]{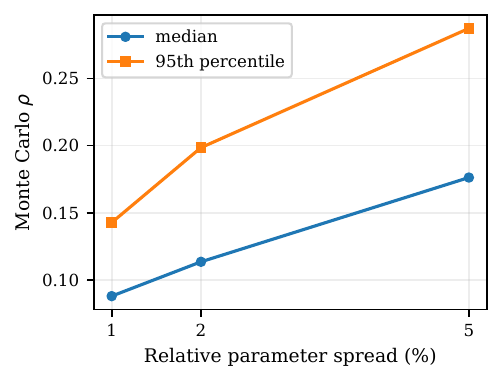}
 \caption{Reproducible Monte Carlo sensitivity of the nominal 11-mode candidate under independent mean-preserving log-normal perturbations.  Solid markers show the median and 95th percentile of the dense-grid maximum.  These statistics are model-dependent diagnostics, not deterministic guarantees.}
 \label{fig:robust}
\end{figure}

For experimental use, the dimensionless design is scaled by a chosen angular-frequency band half-width $B$ with $\kappa=2B$.  The physical mode detunings and amplitude-decay rates (power half-widths at half maximum) are $B\delta_j$ and $B\gamma_j$; the corresponding power full widths at half maximum are $2B\gamma_j$.  Coherent couplings scale as $B\sqrt{w_j}$.  A complete memory characterization should report the complex reflection (or sufficient data for a causal reconstruction), controlled and uncontrolled output ports, the admitted signal subspace, storage time, read protocol, noise occupation, and propagated parameter uncertainty.  Magnitude-only reflection cannot establish minimum phase because all-pass factors leave $|r|$ unchanged.

\section{Stationary versus dynamically controlled memories}
\label{sec:dynamic}

The Bode--Fano result constrains one stationary passive transfer function presented to arbitrary in-band waveforms.  A time-dependent coupler designed for one known incident envelope solves a different problem and does not evade the finite-register obstruction for a single universal linear map on all of $L^2(\B)$.  For a single lossless cavity initially in vacuum, with time-dependent external rate $\kappa(t)$ and a phase control that cancels the input phase chirp, imposing zero output gives the formal schedule
\begin{equation}
 \kappa(t)=\frac{|f(t)|^2}{\displaystyle\int_{-\infty}^{t}|f(\tau)|^2\dd\tau}.
 \label{eq:dynamic-schedule}
\end{equation}
For a pulse with a hard onset or a rapidly decaying leading tail this expression can diverge.  Adding a positive constant to the denominator corresponds to a pre-excited cavity, not a vacuum-memory regularization.  A valid vacuum protocol must instead start in an earlier tail or clip the coupling and account for the resulting reflected norm.  Such mode-specific protocols can approach perfect capture but require waveform knowledge, control bandwidth, timing, and dynamic range \cite{Dilley2012,Greggio2026Optimal}.  They do not contradict \cref{eq:BF-floor}; they abandon stationarity and universality over the full band.

The comparison suggests a resource taxonomy.  Passive multiresonator synthesis spends hardware order to improve an arbitrary-waveform band guarantee.  Dynamic capture spends real-time control and prior waveform information.  Hybrid architectures can use a passive front end to reduce sensitivity and a bounded dynamic coupler to remove residual mismatch.

\section{Scope, limitations, and outlook}

The exact certificate converts a sampled numerical observation into a statement about a fully specified rational system.  Its proof object consists of rational parameters, coefficient hashes, endpoint signs, a Sturm count, and Hurwitz determinant signs; it can be checked without reproducing or trusting the optimization that located the candidate.

The remaining limitations delimit rather than qualify the proved results.  First, fixed-support convexity does not make the movable-pole problem convex.  An order-specific lower bound matching the 11-mode upper certificate would require a dual extremal measure, a valid nonlinear alternation theorem for the constrained family, or another global argument.  Second, \cref{eq:sharp-infimum} settles the limiting value but not the finite-order convergence law.  Root-exponential fits over $N\le19$ remain hypotheses, and optimizer stagnation is visible at $N=19$.  Third, existing multiport Bode--Fano bounds \cite{NieHochwald2017} must be specialized to the quantum controlled-output dilation, the diagonal or matrix positive-residue realization, and a worst-case singular-value objective.  In a genuinely multiport memory, worst-case prompt scattering is controlled by the largest singular value of the reflection matrix, and a rank-deficient storage coupling leaves exactly dark incident directions.  Finally, an actual finite memory must specify a finite-dimensional admitted signal subspace and a physical capture/read protocol; scattering performance alone cannot evade \cref{eq:finite-rank-obstruction}.

These limitations suggest direct extensions.  Interval-valued device parameters can be incorporated into the polynomial certificate by cylindrical algebraic decomposition or sum-of-squares relaxations.  Fixed-support synthesis can be combined with exact a posteriori active-frequency cuts.  Matrix-valued positive-real interpolation can treat polarization, spatial, or frequency-bin ports.  Finally, experimentally measured complex scattering data can be fitted by a passive rational model and passed through the same exact or interval certification pipeline.

\section*{finite-dimensional reduction of the lower certificate}

The boundary-dual analysis of scalar interpolation problems supplies a
concrete reduction for the unresolved finite-order lower bound.

\begin{theorem}[Atomic reduction of a dual lower certificate]
\label{thm:third-memory-atomic-dual}
Assume that a valid positive dual lower certificate for order \(N\) is
specified by its mass, \(d_N\) real moment constraints, and the value of
one continuous objective integrand on a compact frequency domain.  Then a
certificate with the same constraints and objective exists with at most
\(d_N+2\) atoms.  If the dual object is signed, it can be chosen with at
most \(2(d_N+2)\) atoms after preserving the two Jordan masses.
\end{theorem}

\begin{proof}
Apply Carath\'eodory's theorem to the vector of moment functions together
with the objective, after normalizing each positive measure.  For a signed
measure, apply the positive result to its Jordan parts.
\end{proof}

For \(N=11\), this turns the search for a matching lower certificate from
an unrestricted measure problem into a finite system of atom locations,
weights, moment equations, and one global positive-real inequality.
The theorem does not assert that the matching certificate exists, but it
removes diffuse measures as an independent obstruction whenever the
dual formulation has the stated finite-moment form.  Candidate atoms may
be discovered numerically; the final proof must encode their algebraic
values or certified intervals and verify the universal inequality
independently of the optimizer.

\section*{attainment of coercive dual lower certificates}

The finite-atomic reduction yields a complete existence statement for
positive duals whose moments control total mass.

\begin{theorem}[Coercive memory-dual attainment]
\label{thm:frontier-memory-attainment}
Let the compact frequency band be \(X\), and suppose the positive dual
feasible set is cut out by finitely many continuous moment equalities,
one of which has the form
\[
 \int_X f_0\,d\lambda=c_0,\qquad f_0\ge c>0.
\]
If the dual objective is continuous, then its optimum is attained by a
positive measure.  With \(d_N\) remaining real moments, an optimal
certificate can be chosen with at most \(d_N+2\) atoms.
\end{theorem}

\begin{proof}
The coercive equality bounds the mass by \(c_0/c\).  Banach--Alaoglu and
closedness of continuous moment constraints give weak-* compactness.
The objective attains its optimum, and the atomic-reduction theorem
compresses an optimizer without changing its moments or value.
\end{proof}

For the order-\(11\) lower-bound programme this removes both diffuse mass
and nonattainment whenever a coercive positive dual formulation is
established.  The remaining algebraic inequality over all passive designs
is the only model-specific part of that certificate.

\section*{Localization of every optimal dual certificate}

Dual complementarity identifies the frequencies at which an attained
positive certificate can place mass.

\begin{theorem}[Contact-set support theorem]
\label{thm:final-memory-contact}
Consider the attained positive-measure programme of
Theorem~\ref{thm:frontier-memory-attainment} in the form
\[
 \max_{\lambda\ge0}\int_X g\,d\lambda,
 \qquad
 \int_X f_j\,d\lambda=c_j.
\]
Suppose a dual optimum \(y^*\) exists with nonnegative slack
\[
 s_{y^*}(x)=\sum_jy_j^*f_j(x)-g(x)\ge0.
\]
Then every primal optimizer \(\lambda^*\) is supported on the contact set
\[
 Z(y^*)=\{x\in X:s_{y^*}(x)=0\}.
\]
If \(Z(y^*)\) has \(z<\infty\) points, an optimal certificate uses at most
\(\min\{z,d_N+2\}\) atoms.
\end{theorem}

\begin{proof}
Strong duality and feasibility give
\[
 0=y^*\!\cdot c-\int g\,d\lambda^*
 =\int_Xs_{y^*}\,d\lambda^*.
\]
The integrand and measure are nonnegative, so the measure is concentrated
on the zero set.  Intersect this restriction with the atomic bound of
Theorem~\ref{thm:frontier-memory-attainment}.
\end{proof}

The order-\(11\) search can therefore be reduced to the zeros of one
dual slack function.  A finite contact set turns the infinite-frequency
certificate into a finite exact algebraic object.

\section{Conclusion}

We have formulated a passive multiresonator quantum-memory front end as a spectral-admittance problem with an explicit controlled-output dilation and a subspace-qualified capture map.  This identifies the precise condition under which reflection defect becomes reversible memory-write probability.  It also exposes a finite-register obstruction: no finite register can have positive worst-case write efficiency on the full space of arbitrary band-supported waveforms under one universal linear write map.

For stationary diagonal passive-atom front ends, the Bode--Fano floor is the exact infinite-order reflection infimum.  The explicit outer extremal and its positive Herglotz measure prove achievability in the finite-atom closure, while rationality proves strict nonattainment at every finite order.  Fixed pole locations lead to a quasiconvex oscillator-strength subproblem; global finite-order pole placement and its asymptotics remain open.

For the explicit 11-mode candidate in \cref{tab:n11}, exact Sturm and Routh--Hurwitz arithmetic proves uniform reflection below $0.0641023$ and controlled-port transfer above $0.995890895$ on the normalized band.  The certificate is independent of frequency sampling; source code, regression tests, and machine-readable proof data are included.  On a finite-dimensional signal subspace captured isometrically, it is also a write-efficiency guarantee.  Its scope is deliberately precise: it certifies a strong passive front end, not global finite-order optimality or an experimentally complete memory.

\section*{Data and code availability}
The ancillary directory \texttt{anc/} contains the rational design parameters,
deterministic source code, regression tests, and build instructions needed to
regenerate the proof records, numerical tables, and figures.  After installing
\texttt{anc/requirements.txt}, the command
\begin{center}
\small\texttt{cd anc \&\& python3 code/reproduce.py --quick}
\end{center}
regenerates the exact and independent certificate records, runs the regression
tests, and compares all archived figures.  The standard-library cross-check
\begin{center}
\small\texttt{python3 anc/code/independent\_verifier.py --designs anc/data/designs.json}
\end{center}
can be run without the numerical stack.  Generated reports are written only
to a user-created reproduction directory and are not trusted inputs.  No
proprietary or access-restricted data were used.

\section*{Scope, limitations, and open problems}

The exact continuum infimum, passive realization constraints, finite-register
obstruction, and rational finite-order upper certificates are established
within the stated one-port model.  The revised exact run certified every odd
order through \(N=19\), and all six independent core regression tests pass.
The optimization that found a design is not trusted by the certificate.

A principal open problem is a matching finite-order lower bound, beginning with the
fully certified \(N=11\) instance.  A dual extremal measure or a valid
nonlinear alternation theorem would close that gap.  The subsequent
multiport theorem must use matrix positive-real interpolation and the largest
reflection singular value, and must explicitly include the finite admitted
signal subspace and capture/read protocol before the scattering bound can be
interpreted as memory efficiency.

\bibliographystyle{unsrt}
\section*{dual-measure lower certificate}

The finite-order lower-bound target can be put in a proof-carrying form.
For any admissible \(N\)-mode design \(\theta\), write its band defect as
\(f_\theta(\omega)=|r_\theta(i\omega)|^2\).  Every probability measure
\(\lambda\) on the band gives
\[
 \sup_{\omega\in[-1,1]}f_\theta(\omega)
 \ge \int_{-1}^{1}f_\theta(\omega)\,\dd\lambda(\omega).
\]
Consequently, if exact positive-real identities imply
\[
 \inf_{\theta\in\mathcal A_N}
 \int f_\theta\,\dd\lambda \ge L_N,
\]
then \(\lambda\), together with the algebraic proof of the last inequality,
is a global finite-order lower certificate.

For an atomic rational \(\lambda\), the first inequality and all moment
constraints are exact rational statements.  This separates discovery from
proof in the same way as the existing upper certificates: candidate
ripple points may suggest the atoms, but the final certificate must verify
the integral inequality for every positive-residue realization, not only
for the optimized candidate.

The \(N=11\) gap is therefore reduced to one explicit dual object.  The
present revision supplies the certificate format and exact verification
conditions; constructing a measure whose \(L_{11}\) matches the certified
upper value remains the unresolved nonlinear extremal step.

\bibliography{references}

\end{document}